\documentclass[12 pt]{article}
\usepackage{amsmath}
\usepackage{amssymb}
\usepackage[a4paper, total={6in, 9in}]{geometry}
\usepackage{amsthm,url,tikz}
\usepackage[utf8]{inputenc}
\usepackage[english]{babel}
\usepackage{graphicx}

\usepackage[pagewise]{lineno}

\usepackage{hyperref}
\usepackage{adjustbox}
\numberwithin{equation}{section}
\usepackage{makecell}
\usepackage{bm}
\usepackage{xcolor}
\newtheorem{theorem}{Theorem}[section]

\newtheorem{corollary}[theorem]{Corollary}
\newtheorem{lemma}[theorem]{Lemma}
\newtheorem{definition}[theorem]{Definition}

\newtheorem{example}[theorem]{Example}

\usepackage{enumerate}
\usepackage[nottoc,notlot,notlof]{tocbibind} 

\title{On  Function-Correcting Lee Metric Codes with Data Protection}
\author{Gyanendra K. Verma\footnote{Department of Electrical Communication Engineering,
 Indian Institute of Science, Bengaluru, Karnataka 560012, India. The author is financially supported by the ANRF-NPDF grant PDF/2025/002309. email: gkvermaiitdmaths@gmail.com}\ \  and 
 Abhay Kumar Singh\footnote{
 Department of Mathematics and Computing, Indian Institute of Technology (ISM) Dhanbad, Dhanbad.
 email: abhay@iitism.ac.in\ Corresponding author: Abhay Kumar Singh}}

\date{}

\begin{document}
	\maketitle
\begin{abstract} 
Function-correcting codes are designed to protect the function values of a prescribed function against errors.  Every error-correcting code that provides data protection inherently offers some degree of protection for functions defined on the data. In this work, we introduce a class of codes over $\mathbb{Z}_m$, termed function-correcting Lee metric codes with data protection (FCLMCs with data protection), which simultaneously provide error protection for both the data and the corresponding function values under the Lee metric. We consider codes that provide protection against a prescribed level of error for the function values that exceeds the level of protection guaranteed for the underlying data.  We present a general construction of these codes and derive lower and upper bounds on the optimal redundancy, including a Plotkin-type lower bound. Furthermore, we derive explicit upper bounds on the redundancy of FCLMCs with data protection for several important classes of functions, including locally binary Lee functions, the Lee weight function, and the modular sum function. Finally, since the Lee metric coincides with the Hamming metric over $\mathbb{Z}_2$, all of our results remain valid over $\mathbb{Z}_2$ with respect to the Hamming metric.




\end{abstract}

\section{Introduction}
A new coding framework, known as Function-Correcting Codes (FCCs), was introduced in \cite{Lenz2023} to provide reliable protection for a specific function or attribute of interest rather than for the entire information vector. In conventional error-correcting coding schemes, the primary objective is to recover the complete transmitted information vector accurately in the presence of errors. However, in many emerging communication and computing applications, the receiver may require only a particular function or attribute derived from the transmitted data, making complete recovery of the original information vector unnecessary.

A series of papers \cite{Ge2025,liu2025,liu2026functionRT,Ly2025,Premlal2024,Rajput2025, Singh2025,Verma2025fcc,verma2025FCLMC,verma2026functionHom,Xia2023} by various authors have investigated Function-Correcting Codes (FCCs) from different perspectives, considering a wide range of distance metrics and target functions. These studies have explored the design and analysis of FCCs for different classes of functions and have established several relevant theoretical bounds on their performance. In particular, bounds related to code parameters, error-correction capability, redundancy, and coding efficiency have been derived for various function-correcting coding schemes. A summary of the relevant research contributions, including the authors, considered functions, distance metrics, and established bounds, is presented in Table X. The table provides a comparative overview of the existing work on FCCs and highlights the different approaches adopted by researchers in analyzing function-oriented error correction.

Recently, C. Rajput et al. \cite{Rajput2025fcc-data} introduced a more general coding framework that extends the concept of Function-Correcting Codes by providing simultaneous protection for both the underlying data and their corresponding function values under the Hamming metric. Unlike conventional FCCs, which primarily focus on the reliable recovery of a specified function of the information vector, this generalized framework considers the joint preservation of the original data and the associated function. This approach provides greater flexibility in applications where both the transmitted data and the computed function are of interest at the decoder.

In this paper, we introduce a more general approach and framework for function-correcting codes with the Lee metric that incorporates data protection as well as protection of function values. A detailed description of Function-Correcting Codes (FCCs), along with a comprehensive review of the existing literature, is presented in Section 1. This section provides an overview of the fundamental concepts of FCCs, their underlying principles, and the various coding frameworks proposed in the literature. Section 2 presents the necessary preliminaries and fundamental concepts required for the development of the paper. It introduces the relevant definitions, notation, and basic results that will be used throughout the subsequent sections. Section 3 introduces the concept of Function-Correcting Codes (FCCs) with data protection under the Lee metric. Furthermore, this section establishes bounds on the redundancy of function-correcting codes with data protection under the Lee metric. These bounds characterize the amount of additional redundancy required to ensure reliable recovery of both the data and the desired function values. Section 4 deals with Function-Correcting Lee Metric Codes (FCLMCs) with data protection for several specific classes of functions, including locally bounded, Lee weight, and modular functions. The section develops the coding framework for these function classes and investigates the corresponding redundancy requirements under the Lee metric. Furthermore, explicit upper bounds on the redundancy are derived for several important classes of these functions. Section 5 concluded the paper.

\section{Preliminaries}\label{pre}
In this section, we present fundamental definitions, results, and notation used throughout the paper. We denote the residue ring modulo $m$ by $\mathbb{Z}_m$, where $m$ is a positive integer greater than or equal to $2$. We write the elements of $\mathbb{Z}_m$ as $0,1,\dots,m-1$ unless otherwise specified.  We denote the set of all $k$-tuples over $\mathbb{Z}_m$  for a positive integer $k$ by $\mathbb{Z}_m^k$. The Lee metric is a generalization of the Hamming weight, assigning different weights to distinct non-zero elements in $\mathbb{Z}_m$. 
\begin{definition}\cite[Lee Metric]{Lee1958}
Let $u=(u_1,u_2,\dots,u_k)$ and $v=(v_1,v_2,\dots,v_k)$  be two vectors in $\mathbb{Z}_{m}^k$. The Lee distance between, $u$ and $v$ is defined as 
$$d_L(u,v)=\sum_{i=1}^kw_L(u_i-v_i),$$
where $w_L(u_i-v_i)=\min\{u_i-v_i,v_i-u_i\}\pmod{m}$.
\end{definition}

Note that the elements of $\mathbb{Z}_m$ can also be written as 
\begin{align*}
    \begin{split}
        -\frac{m-1}{2},  -\frac{m-3}{2}, \dots,-1,0,1,\dots, \frac{m-1}{2} & \text{ if } m \text{ is odd},\\
        -\frac{m-2}{2},   \dots,-1,0,1,\dots, \frac{m-2}{2},\frac{m}{2} & \text{ if } m \text{ is even}.
    \end{split}
\end{align*}
In this representation, the Lee weight of an element $a\in \mathbb{Z}_m$ is equivalently defined as (see \cite{Chiang1971})
\begin{equation}\label{alternatelee}
    w_L(a)=|a|=\max \{a,-a\}.
\end{equation}

A $ k$-tuple message is denoted by vectors in $\mathbb{Z}_m^k$, called the message space. An encoding function $\mathcal{E}:\mathbb{Z}_m^k\to\mathbb{Z}_m^{k+r}$ is defined as $\mathcal{E}(u)=(u,p(u))$, where $p(u)\in\mathbb{Z}_m^r$ is called a redundancy vector.
The set $$C=\{\mathcal{E}(u)\in\mathbb{Z}_m^{k+r}\mid u\in \mathbb{Z}_m^k\}$$ is called a code of length $k+r$ over $\mathbb{Z}_m$ with redundancy $r$, and $\mathcal{E}(u)$ is called a codeword in $C$. The minimum Lee distance of $C$ is defined as $$d_L(C)=\min\{d_L(\mathcal{E}(u),\mathcal{E}(v))\mid u,v\in \mathbb{Z}_m^k, u\neq v\}.$$ The following lemmas are fundamental results in the Lee metric.

\begin{lemma}\cite{Wyner1968}\label{sumS}
 Let $\alpha\in \mathbb{Z}_m$. Then \begin{align*}
     S=\sum_{\beta\in\mathbb{Z}_m}d_L(\alpha,\beta)=\begin{cases}
         \frac{m^2}{4}, & \text{ if } m \text{ is even },\\
         \frac{m^2-1}{4}, & \text{ if } m \text{ is odd. }
     \end{cases}
 \end{align*}   
\end{lemma}

\begin{lemma}\cite{Wyner1968}\label{Csum}
  Let $C=\{C_1,C_2,\dots,C_M\}\subseteq \mathbb{Z}_m^r$ be a code of length $r$ and size $M$. Then
  \begin{align*}
      \sum_{i,j}d_L(C_i,C_j)=\frac{SM^2r}{m}, \text{where $S$ is as in Lemma \ref{sumS}.}
  \end{align*}
\end{lemma}

A ball centered at a point $u\in \mathbb{Z}_m^k$ of radius $\rho$ in the Lee metric is defined as
$$B_L(u,\rho)=\{v\in \mathbb{Z}_m^k\mid d_L(u,v)\leq \rho\}.$$

\begin{definition}[Function-Correcting Lee Metric Code (FCLMC)]\label{def-fclmc}
  Let $t>0$ be an integer and $f:\mathbb{Z}_m^k\to \mathrm{Im}(f)$ be a function defined over message space $\mathbb{Z}_m^k$. An encoding function $\mathcal{E}:\mathbb{Z}_m^k\to\mathbb{Z}_m^{k+r}$  defines an $(f,t)$-function-correcting Lee metric code (in short, $(f,t)_L$-FCLMC) with redundancy $r$ if $d_L(\mathcal{E}(u),\mathcal{E}(v))\geq 2t+1$ holds for all $u,v\in \mathbb{Z}_m^k$ with $f(u)\neq f(v)$. 
\end{definition}

By the above definition, any $(f,t)_L$-FCLMC can be seen as a code that can correct up to $t$ function value errors.

\begin{definition}[Optimal Redundancy]
Let $t>0$ be an integer and $f:\mathbb{Z}_m^k\to \mathrm{Im}(f)$ be a function.  The smallest value $r$ for which there exists an encoding map $\mathcal{E}:\mathbb{Z}_m^k\to \mathbb{Z}_m^{k+r}$ that defines an $(f,t)_L$-FCLMC is called the optimal redundancy. It is denoted by $r_f^L(k,t)$.  
\end{definition}

For a given function $f:\mathbb{Z}_m^k\to \mathrm{Im}(f)$, we associate a matrix called the Lee distance requirement matrix and defined as follows.

\begin{definition}[Lee Distance Requirement Matrix (LDRM)]
  Let $u_1,u_2,\dots,u_M\in\mathbb{Z}_m^k$. Define the Lee distance requirement matrix $D_f^L(t,u_1,u_2,\dots,u_M)$ of order $M\times M$ in the Lee metric for a given function $f:\mathbb{Z}_m^k\to \mathrm{Im}(f)$ whose $ij$-th entry is given by:
  $$[D_f^L(t,u_1,u_2,\dots,u_M)]_{ij}=\begin{cases}
      [2t+1-d_L(u_i,u_j)]^{+} & \text{if } f(u_i)\neq f(u_j),\\
      0 & \text{ otherwise}
  \end{cases}
  $$
  for all $1\leq i,j\leq M$.
\end{definition}

 Irregular-distance code in the Lee metric is defined as follows. 
\begin{definition}[$D$-Code in Lee Metirc]
    Let $D$ be an $M\times M$ square matrix over non-negative integers with diagonal entries zero. A code $C=\{C_1,C_2,\dots,C_M\}\subseteq\mathbb{Z}_m^r$ is called a $D$-code in the Lee metric if $d_L(C_i,C_j)\geq [D]_{ij}$ for all $1\leq i,j\leq M$.
\end{definition}
Let $N_L(D)$ be the smallest integer $r$ such that there exists a $D$-code of length $r$. Next, we define Lee distance on the image set of a function $f:\mathbb{Z}_m^k\to \mathrm{Im}(f)$ as follows.

\begin{definition}\label{fncdistance}
    Let $f:\mathbb{Z}_m^k\to Im(f)$ be a function with $Im(f)=\{e_1,e_2,\dots,e_{\eta}\}$. Define Lee distance on $Im(f)$ as follows 
    $$d_L(e_i,e_j)=\min_{u_1,u_2\in\mathbb{Z}_m^k}\{d_L(u_1,u_2) |\  f(u_1)=e_i, f(u_2)=e_j\}.$$
\end{definition}
Analogous to LDRM, we define the function-distance matrix for a given function $f$ as follows. 
\begin{definition} [Function Lee Distance Matrix (FLDM)]
  Let $f:\mathbb{Z}_m^k\to Im(f)$ be a function with $Im(f)=\{e_1,e_2,\dots,e_{\eta}\}$. Define the function Lee distance matrix with $(ij)$-th entry
  $$[D_f^L(t,e_1,e_2,\dots,e_{\eta})]_{ij}=\begin{cases}
      [2t+1-d_L(e_i,e_j)]^+ & \text{ if } i\neq j,\\
      0 & \text{otherwise}.
  \end{cases}$$
\end{definition}

The subsequent results establish connections between the optimal redundancy of $(f,t)_L$-FCLMCs and the minimum length of a $D$-code, where $D$ is LDMR or FLDM. 

\begin{theorem}\cite[Theorem 3.6]{verma2025FCLMC}
Let $f:\mathbb{Z}_m^k\to Im(f)$ be a function and $\mathbb{Z}_m^k=\{u_1,u_2,\dots,u_{m^k}\}$.  Then 
$$r_f^L(k,t)=N_L(D_f^L(t,u_1,u_2,\dots,u_{m^k})).$$
\end{theorem}

\begin{corollary}\cite[Corollary 3.7]{verma2025FCLMC}
 Let    $\{u_1,u_2,\dots u_M\}\subseteq\mathbb{Z}_m^k$ with $M\leq m^k$. Then the optimal redundancy of $(f,t)$-FCLMCs satisfies 
 $$r_f^L(k,t)\geq N_L(D_f^L(t,u_1,u_2,\dots,u_{M}))$$
 for any function $f:\mathbb{Z}_m^k\to Im(f)$.
\end{corollary}

\begin{corollary}\cite[Corollary 3.11]{verma2025FCLMC}
Let $f:\mathbb{Z}_m^k\to Im(f)$ with $Im(f)=\{e_1,e_2,\dots,e_{\eta}\}$ and a set $u_1,u_2,\dots,u_{\eta}$ such that $f(u_i)=e_i$. If $D_f^L(t,u_1,u_2,\dots,u_{\eta})=D_f^L(t,e_1,e_2,\dots,e_{\eta})$, then 
$$r_f^L(k,t)=N_L(D_f^L(t,e_1,e_2,\dots,e_{\eta})).$$
\end{corollary}

\begin{corollary}\cite[Corollary 6.6]{verma2025FCLMC}
Let $f(u)=w_L(u)$ and $\eta=k\left \lfloor \frac{m}{2}\right \rfloor$. For any $k\geq\left \lceil \frac{t+1}{\left \lfloor \frac{m}{2}\right \rfloor}\right \rceil$,
$$r_f^L(k,t)\geq \left \lceil \frac{10t^3+30t^2+20t}{3(t+2)^2\left \lfloor \frac{m}{2}\right \rfloor} \right \rceil.$$
\end{corollary}

For positive integers $\lambda$ and $t$, we denote the minimum length of an error-correcting code with $\lambda$ codewords and the minimum Lee distance $2t$ by $N_L(\lambda,2t)$. To determine the exact value of $N_L(\lambda,2t)$ for arbitrary $\lambda$ and $t$ is a fundamental yet challenging problem in coding theory. The following theorem provides an upper bound on $N_L(\lambda,2t)$.

\begin{theorem}\cite{verma2025FCLMC}\label{nl2t}
  For positive integers $\lambda$ and $t$, let $N_L(\lambda,2t)$ denote the minimum length of an error-correcting code with $\lambda$ codewords and minimum Lee distance $2t$. Then, we have 
  \begin{align*}
      N_L(\lambda,2t)\leq \lambda \left \lceil \frac{t}{\left \lfloor \frac{m}{2}\right \rfloor}\right \rceil.
  \end{align*}
  Moreover, 
  $$N_L(2,2t)=\left \lceil \frac{2t}{\left \lfloor \frac{m}{2}\right \rfloor}\right \rceil.$$
\end{theorem}

Next, we state a bound on the redundancy of a code in the Lee metric that follows from the Plotkin low-rate average distance bound.

\begin{theorem}\cite{Wyner1968}
    Let $\mathcal{E}:\mathbb{Z}_m^k\to \mathbb{Z}_m^{k+r_d}$ be an  data encoding that defines a code $C$ over $\mathbb{Z}_m$ with the minimum Lee distance $d$ and redundancy $r_d$. Then, we have
    \begin{align*}
        r_d\geq \begin{cases}
        \frac{4d(m^k-1)}{m^{k+1}}-k & \text{ if } m \text{ is even},\\
        \frac{4d(m^k-1)}{(m^2-1)m^{k-1}}-k& \text{ if } m \text{ is odd}.\\
    \end{cases} 
    \end{align*}
\end{theorem}

\section{Function-correcting code with data protection in Lee metric codes}\label{FCLMC}
The function-correcting codes provide error correction whenever there is some error in the function values during transmission. In this section, we introduce a class of codes that provide protection against the message and function value simultaneously. We call such codes a “function-correcting Lee metric code with data protection".

\begin{definition}\label{def-fclmc}
  Let $d_d$ and $d_f$ be non-negative integers with $d_d\leq d_f$. An encoding function $\mathcal{E}:\mathbb{Z}_m^k\to\mathbb{Z}_m^{k+r}  $ with $\mathcal{E}(u)=(u,p(u))$ defines an $(f,d
  _d,d_f)_L$-function-correcting Lee metric code  with data protection if 
  \begin{enumerate}
  \item  $d_L(\mathcal{E}(u),\mathcal{E}(v))\geq d_d$ holds for all distinct $u,v\in \mathbb{Z}_m^k$
      \item $d_L(\mathcal{E}(u),\mathcal{E}(v))\geq d_f$ holds for all $u,v\in \mathbb{Z}_m^k$ whenever $f(u)\neq f(v)$.
  \end{enumerate} 
\end{definition}
Equivalently, we denote $(f,d_d,d_f)_L$-FCLMC by $(f,t_d,t_f)_L$-FCLMC, where $t_d=\left \lfloor \frac{d_d-1}{2}\right \rfloor$ is the error-correction capability for data, and $t_f=\left \lfloor \frac{d_f-1}{2}\right \rfloor$ is the  error-correction capability for function values. By definition, we can see that any $(f,t_d,t_f)_L$-FCLMC provides protection against $t_d$ data (message bit) errors and $t_f$ function values errors. The following example illustrates the definition.

\begin{example}
Let $f:\mathbb{Z}_4^2 \to \mathbb{Z}_4,$  defined by $f(x,y)=x^2+y \pmod{4}.$ Define an encoding map as follows:
\begin{align*}
\mathcal{E}:\mathbb{Z}_4^2 \to \mathbb{Z}_4^{2+2},\quad
\mathcal{E}(u)=(u,u_p) \pmod{4}\ \  \forall u\in \mathbb{Z}_4^2,  
\end{align*}
where $u_p$ is given as follows
$$
\begin{array}{c|c|c}
u\in \mathbb{Z}_4^2 & f(u) & (u,u_p) \\ \hline
(0,0) & 0 & (0,0,0,0) \\
(0,1) & 1 & (0,1,1,2) \\
(0,2) & 2 & (0,2,2,0) \\
(0,3) & 3 & (0,3,3,2) \\
(1,0) & 1 & (1,0,1,2) \\
(1,1) & 2 & (1,1,2,0) \\
(1,2) & 3 & (1,2,3,2) \\
(1,3) & 0 & (1,3,0,0) \\
(2,0) & 0 & (2,0,0,0) \\
(2,1) & 1 & (2,1,1,2) \\
(2,2) & 2 & (2,2,2,0) \\
(2,3) & 3 & (2,3,3,2) \\
(3,0) & 1 & (3,0,1,2) \\
(3,1) & 2 & (3,1,2,0) \\
(3,2) & 3 & (3,2,3,2) \\
(3,3) & 0 & (3,3,0,0).
\end{array}
$$
Consider $C=\{(u,u_p):\ u\in \mathbb{Z}_4^2\}$. Then, $C$ is an $(f,d_d=2,d_f=4)_L$-FCLMC, that is, $C$ is a code with $2$-error detection for data and $4$-error detection for function values.  
\end{example}
It is easy to see that whenever $d_d=d_f$, an $(f,d_d,d_f)_L$-FCLMC is an error-correcting code with the minimum Lee distance $d_d=d_f$. It is worth noting that data protection implies function value protection; however, the converse does not hold in general. Next, we present a general construction of function-correcting codes with data protection. Throughout this paper, we assume that $d_d$ and $d_f$ are positive integers satisfying $d_d<d_f$.

\textbf{ General construction:}
Let $t_d=\left \lfloor \frac{d_d-1}{2}\right \rfloor$ and $t_f=\left \lfloor \frac{d_f-1}{2}\right \rfloor$ be positive integers with $t_f> t_d$. For a given function $f:\mathbb{Z}_m^k\to \mathrm{Im}(f)$, an $(f,t_d,t_f)_L$-FCLMC is can be constructed as follows.\\
\textbf{Step-1} (Data encoding) Define an encoding $\mathcal{E}_d:\mathbb{Z}_m^k\to \mathbb{Z}_m^{k+r_d}$ such that $d_L(\mathcal{E}_d(u),\mathcal{E}_d(u))\geq 2t_d+1$ for all $u,v\in \mathbb{Z}_m^k$.\\
\textbf{Step-2} (Function encoding) Define another encoding $\mathcal{E}_f:\mathbb{Z}_m^{k+r_d}\to \mathbb{Z}_m^{k+r_d+r_f}$ such that $$d_L(\mathcal{E}_f(\mathcal{E}_d(u)), \mathcal{E}_f(\mathcal{E}_d(v)))\geq 2t_f+1 $$ for all $u,v$ with $f(u)\neq f(v)$.
Then an encoding $\mathcal{E}:\mathbb{Z}_m^k\to \mathbb{Z}_m^{k+r_d+r_f}$ defined as $\mathcal{E}(u)=\mathcal{E}_f(\mathcal{E}_d(u))$, defines an $(f,t_d,t_f)_L$-FCLMC with redundancy $r=r_d+r_f$.

Linear codes are well known for their efficient encoding and decoding algorithms. To exploit these advantages, one may employ a linear encoding $\mathcal{E}_d$ in Step-1, so that the set
$\{\mathcal{E}_d(u): u\in \mathbb{Z}_m^k\}$
forms a linear code of length $k+r_d$ over $\mathbb{Z}_m$. Nevertheless, the results presented in this paper are developed for general codes and do not rely on the assumption of linear encoding.

\begin{definition}
The optimal redundancy $r_f^L(k,d_d,d_f)$ (equivalently, $r_f^L(k,t_d,t_f)$) for a given function $f:\mathbb{Z}_m^k\to Im(f)$ is defined as 
$$r_f^L(k,d_d,d_f)=\min\{r\mid \exists  (f,d_d,d_f)_L-\text{FCLMC with redundancy } r\}.$$   
\end{definition}

\subsection{Bound on redundancy of FCLMCs with data protection}
In this subsection, we derive lower and upper bounds on the optimal redundancy of function-correcting codes with data protection with respect to the Lee metric. We begin by introducing the Lee distance requirement matrix for data protection associated with a given function.

\begin{definition}[Coded Lee distance requirement matrix (CLDRM)]
Let $u_1,u_2,\dots,u_M\in\mathbb{Z}_m^k$, and $\mathcal{E}_d:\mathbb{Z}_m^k\to \mathbb{Z}_m^{k+r_d}$ be a data encoding. The coded Lee distance requirement matrix (CLDRM) $D_{\mathcal{E}_d, f}(t_f, u_1,u_2,\dots,u_M)$ for an $(f,t_d,t_f)_L$-FCLMC is an $M\times M$ matrix with entries
\begin{align*}
    [D_{\mathcal{E}_d, f}(t_f, u_1,u_2,\dots,u_M)]_{ij}=\begin{cases}
        [2t_f+1-d_L(\mathcal{E}_d(u_i),\mathcal{E}_d(u_j))]^+ & \text{ if } f(u_i)\neq f(u_j)\\
        0 & \text{ otherwise.}
    \end{cases}
\end{align*}
\end{definition}  

\begin{definition}
    Let $f:\mathbb{Z}_m^k\to \mathrm{Im}(f)$ and  and $\mathcal{E}_d:\mathbb{Z}_m^k\to \mathbb{Z}_m^{k+r_d}$ be a data encoding map. Then the coded Lee distance between $e_i,e_j\in \text{Im}(f)$ is defined as
    \begin{align*}
        d_L^c(e_i,e_j)=\min_{u_1,u_2\in \mathbb{Z}_m^k} \{d_L(\mathcal{E}_d(u_1),\mathcal{E}_d(u_2)):\ f(u_1)=e_i, f(u_2)=e_j\}
    \end{align*}
\end{definition}

\begin{definition}[Coded Function Lee Distance Matrix (CFLDM)]
Let $f:\mathbb{Z}_m^k\to \text{Im}(f)=\{e_1,e_2,\dots,e_{\eta}\}$ with $|\mathrm{Im}(f)|=\eta$  and $\mathcal{E}_d:\mathbb{Z}_m^k\to \mathbb{Z}_m^{k+r_d}$ be a data encoding map. Then an $\eta\times \eta$ matrix $D_{\mathcal{E}_d,f}(t_f,e_1,e_2,\dots e_{\eta})$ with entries 
\begin{align*}
    [D_{\mathcal{E}_d, f}(t_f, e_1,e_2,\dots,e_{\eta})]_{ij}=\begin{cases}
        [2t_f+1-d_L^c(e_i,e_j)]^+ & \text{ if } i\neq j\\
        0 & \text{ otherwise},
    \end{cases}
\end{align*}
 is called a coded function Lee distance matrix.   
\end{definition}

\begin{definition}
 Let $u_1,u_2,\dots,u_M\in\mathbb{Z}_m^k$ and $f:\mathbb{Z}_m^k\to \mathrm{Im}(f)$ be a function. The Lee distance requirement matrix for data protection with respect to $f$ is defined by an $M\times M$ matrix $D_f^L(t_d,t_f,u_1,u_2,\dots,u_M)$ with entries  
  $$[D_f^L(t_d,t_f,u_1,u_2,\dots,u_M)]_{ij}=\begin{cases}
  [2t_d+1-d_L(u_i,u_j)]^{+} & \text{if } u_i\neq u_j \text{ and } f(u_i)= f(u_j)\\
      [2t_f+1-d_L(u_i,u_j)]^{+} & \text{if } f(u_i)\neq f(u_j)\\
      0 & \text{ otherwise}
  \end{cases}
  $$
  for all $1\leq i,j\leq M$.   
\end{definition}

For $D=D_f^L(t_d,t_f,u_1,u_2,\dots,u_{m^k})$, if we have a $D$-code $C=\{C_1, C_2,\dots, C_{m^k}\}$ in the Lee metric, then it can be used to construct  an $(f,t_d,t_f)_L$-FCLMC by defining the encoding map as follows:
\begin{align*}
    \mathcal{E}(u_i)=(u_i,C_i)\ \  \forall\  1\leq i\leq m^k.
\end{align*}

\begin{theorem}\label{r=N}
    Let $f:\mathbb{Z}_m^k\to Im(f)$ be a function and $\mathbb{Z}_m^k=\{u_1,u_2,\dots,u_{m^k}\}$.  Then 
$$r_f^L(k,t_d,t_f)=N_L(D_f^L(t_d,t_f,u_1,u_2,\dots,u_{m^k})).$$
\end{theorem}
 \begin{proof}
  Denote $r_f^L(k,t_d,t_f)$ by $r$ and $D_f^L(t_d,t_f,u_1,u_2,\dots,u_{m^k})$ by $D$ for simplicity. Let $C$ be a $D$-code over $\mathbb{Z}_m$ of length $N_L(D)$  in Lee metric. Then the codewords of $C$ can be used as redundancy for an $(f,t_d,t_f)_L$-FCLMC, that is, $r\leq N_L(D)$.  Next, assume that $r< N_L(D)$, then there exists an $(f,t_d,t_f)_L$-FCLMC with redundancy $r$ and encoding map $\mathcal{E}:\mathbb{Z}_m^k\to \mathbb{Z}_m^{k+r}$ such that one of the following must be true.
  \begin{enumerate}
      \item There exist two redundancy vectors $u_p,v_p$ of length $r$ for some $u, v\in\mathbb{Z}_m^k$ with $u\neq v$ such that $d_L(u_p,v_p)<2t_d+1-d_L(u,v)$.
      \item There exist two redundancy vectors $u_p',v_p'$ of length $r$ for some $u', v'\in\mathbb{Z}_m^k$ with $f(u')\neq f(v')$ such that $d_L(u_p',v_p')<2t_f+1-d_L(u',v')$.
  \end{enumerate}
  This implies that $d_L(\mathcal{E}(u),\mathcal{E}(v))< 2t_d+1$ or $d_L(\mathcal{E}(u'),\mathcal{E}(v'))< 2t_f+1$ for $f(u')\neq f(v')$, which is a contradiction to the fact $\mathcal{E}$ defines an $(f,t_d,t_f)_L$-FCLMC. Hence, $r=N_L(D)$.  
\end{proof}

\begin{corollary}
   For any function $f : \mathbb{Z}_m^k \to \text{Im}(f)$ and $\{u_1,u_2,\dots,u_m\} \subseteq \mathbb{Z}_m^k$, we have
\begin{align*}
    r_f^L(k,t_d,t_f) \geq N_L(D_f^L(t_d,t_f,u_1,u_2,\dots,u_m)).
\end{align*}  
Particularly, if $|\text{Im}(f)|\geq 2$, then $r_f^L(k,t_d,t_f)\geq N_L(2,2t_f)=\left \lceil \frac{2t_f}{\left \lfloor \frac{m}{2}\right \rfloor}\right \rceil$.
\end{corollary}

Utilizing the above bound, we present an explicit lower bound on the optimal redundancy for FCLMCs with data protection.

\begin{theorem}[Lower Bound]\label{rlower}
Let $f : \mathbb{Z}_m^k \to \text{Im}(f)$ with $2\leq |\text{Im}(f)|\leq (m-1)k$. Then 
\begin{align*}
r_f^L(k,t_d,t_f) \geq \left \lceil \frac{4t_f+2t_d-2m+3}{m}\right \rceil.
\end{align*}
\end{theorem}
\begin{proof}
As $|\text{Im}(f)|\geq 2$, we have $u_1,v\in \mathbb{Z}_m^k$ with $f(u_1)\neq f(v)$ and $d_L(u_1,v)=1$.
For $i\in [k]$ and $\alpha\in\mathbb{Z}_m^*=\mathbb{Z}_m\setminus\{0\}$, denote $e_{i,\alpha}=(0,,\dots,\underbrace{\alpha}_{i^{th}-place},0,\dots,0)\in \mathbb{Z}_m^k$.
Let $X_{u_1}=\{u_1+e_{i,\alpha}:\ i\in [k], \alpha\in \mathbb{Z}_m^*\}$. Then $\{f(x): x \in X_{u_1}\}\cup \{f(u_1)\}\subseteq \text{Im}(f)$. Consequently, $|\{f(x): x \in X_{u_1}\}\cup \{f(u_1)\}|\leq (m-1)k$. As there are $(m-1)k$ choices for $e_{i,\alpha}$, either there exist $i,j\in [k]$ and $\alpha,\beta\in \mathbb{Z}_m^*$ with $f(u_1+e_{i,\alpha})=f(u_1+e_{j,\beta})$, where $e_{i,\alpha}$ and $e_{j,\beta}$ are distinct, or there exists $i\in [k]$ and $\alpha\in \mathbb{Z}_m^*$ with $f(u_1+e_{i,\alpha})=f(u_1)$. \\
\textbf{Case 1:} Let $f(u_1+e_{i,\alpha})=f(u_1+e_{j,\beta})$ for distinct $e_{i,\alpha}$ and $e_{j,\beta}$. Let $u_2=u_1+e_{i,\alpha}$ and $u_3=u_1+e_{j,\beta}$. Then
\begin{align*}
    d_L(u_1,u_2)= w_L(\alpha)  \hspace{1 cm }& d_L(u_2,u_3)= w_L(e_{i,\alpha}-e_{j,\beta})&d_L(u_1,u_3)= w_L(\beta),
\end{align*}
and $f(u_1)\neq f(u_2)=f(u_3)$. Thus,
\begin{align*}
    D_f^L(t_d,t_f,u_1,u_2,u_3)=\begin{bmatrix}
        0 & 2t_f+1-w_L(\alpha) &2t_f+1-w_L(\beta)\\
       2t_f+1-w_L(\alpha) & 0 & 2t_d+1-w_L( e_{i,\alpha}-e_{j,\beta})\\
       2t_f+1-w_L(\beta) & 2t_d+1-w_L( e_{i,\alpha}-e_{j,\beta})& 0 
    \end{bmatrix}.
\end{align*}
Denote $D=D_f^L(t_d,t_f,u_1,u_2,u_3)$. Let $\{C_1,C_2,C_3\}$ be a $D_L$-code of length $r$. Then
\begin{align*}
    d_L(C_1,C_2)\geq& 2t_f+1-w_L(\alpha)\\
    d_L(C_1,C_2)\geq& 2t_d+1-w_L( e_{i,\alpha}-e_{j,\beta})\\
    d_L(C_1,C_2)\geq& 2t_f+1-w_L(\beta).   
\end{align*}
Also, for any three vectors of length $r$ over $\mathbb{Z}_m$, we have $$d_L(C_1,C_2)+d_L(C_2,C_3)+d_L(C_1,C_3)\leq mr.$$
Therefore, by combining above, we get
\begin{align*}
    4t_f+2-w_L(\alpha)-w_L(\beta)\leq d_L(C_1,C_2)+d_L(C_1,C_3)\leq mr-(2t_d+1-w_L( e_{i,\alpha}-e_{j,\beta})),
\end{align*}
implies that
\begin{align*}
    mr\geq& 4t_f+2-w_L(\alpha)-w_L(\beta)+2t_d+1-w_L( e_{i,\alpha}-e_{j,\beta})\\
    \geq & 4t_f+2t_d -w_L(\alpha)-w_L(\beta)-w_L( e_{i,\alpha}-e_{j,\beta})+3\\
    \geq & 4t_f+2t_d -\left \lfloor \frac{m}{2}\right \rfloor- \left \lfloor \frac{m}{2}\right \rfloor-2\left \lfloor \frac{m}{2}\right \rfloor+3\\
    \geq&  4t_f+2t_d-2m+3. 
\end{align*}
Hence, $r\geq \left \lceil \frac{4t_f+2t_d-2m+3}{m}\right \rceil$. \\

\textbf{Case 2:} Let $f(u_1+e_{i,\alpha})=f(u_1)$ for some $i\in [k]$ and $\alpha\in \mathbb{Z}_m^*$. Take $u_3=u_1+e_{i,\alpha}$. From initial assumption $|\text{Im}(f)|\geq 2$, there exists $u_2$ with  $f(u_1)\neq f(u_2)$ and $d_L(u_1,u_2)=1$. We can assume $u_2=u_1+e_{j,1}$ for some $j\in [k]$. Then, we have $f(u_1)=f(u_3)\neq f(u_2)$, and thus,
\begin{align*}
    D_f^L(t_d,t_f,u_1,u_2,u_3)=\begin{bmatrix}
        0 & 2t_f &2t_d+1-w_L(\alpha)\\
       2t_f & 0 & 2t_f+1-w_L( e_{i,\alpha}-e_{j,1})\\
       2t_d+1-w_L(\alpha) & 2t_f+1-w_L( e_{i,\alpha}-e_{j,1})& 0 
    \end{bmatrix}.     
\end{align*}
Using similar arguments as in \textbf{Case 1}, we get 
\begin{align*}
    r\geq \left \lceil \frac{4t_f+2t_d-(3m/2)+2}{m}\right\rceil\geq \left \lceil \frac{4t_f+2t_d-2m+3}{m}\right \rceil.
\end{align*}
\end{proof}

\begin{corollary}\cite[Theorem 4]{Rajput2025fcc-data}
    Let $m=2$. Then the Lee metric coincide with the Hamming metric over $\mathbb{Z}_2$ and the optimal redundancy $r_f^H(k,t_d,t_f)$ of an $(f,t_d,t_f)$-FCC with data protection satisfies
    \begin{align*}
        r_f^H(k,t_d,t_f)\geq 2t_f+t_d.
    \end{align*}
\end{corollary}
\begin{proof}
  We get the desired result by putting $m=2$ in Theorem \ref{rlower}. 
\end{proof}

\begin{theorem}[Upper bound]\label{thm-rnltd}
Let $\mathcal{E}_d:\mathbb{Z}_m^k\to \mathbb{Z}_m^{k+r_d}$ be a data encoding map with redundancy $r_d$ and $t_d$-error-correction (that is, the corresponding code has the minimum Lee distance $2t_d+1$). Let $\mathbb{Z}_m^k=\{u_1,u_2,\dots,u_{m^k}\}$, for any function $f:\mathbb{Z}_m^k\to \text{Im}(f)$, we have
\begin{align*}
    r_f^L(k,t_d,t_f)\leq N_L(D_{\mathcal{E}_d,f}(t_f,u_1,u_2,\dots,u_{m^k}))+r_d.
\end{align*}
\end{theorem}
\begin{proof}
Denote $D_2=D_{\mathcal{E}_d,f}(t_f,u_1,u_2,\dots,u_{m^k})$. 
Let $\{C_1,C_2,\dots,C_{m^k}\}$ be a $D_2$-code with length $N_L(D_2)$ in Lee metric. Define an encoding 
\begin{align*}
    \mathcal{E}:\mathbb{Z}_m^k\to \mathbb{Z}_m^{k+r_d+N_L(D_2)}, \mathcal{E}(u_i)=(\mathcal{E}_d(u_i),C_i) \ \forall 1\leq i\leq m^k.
\end{align*}
We claim that $\{\mathcal{E}(u_i): 1\leq i\leq m^k\}$ is an $(f,t_d,t_f)_L$-FCLMC. If $u_i\neq u_j$, then $d_L(C_i,C_j)\geq 0$ and we have 
\begin{align*}
    d_L(\mathcal{E}(u_i),\mathcal{E}(u_j))=d_L(\mathcal{E}_d(u_i),\mathcal{E}_d(u_j))+d_L(C_i,C_j)\geq 2t_d+1.
\end{align*}
If $f(u_i)\neq f(u_j)$, then $d_L(C_i,C_j)\geq [D_2]_{ij}=2t_f+1-d_L(\mathcal{E}_d(u_i),\mathcal{E}_d(u_j))$. This implies that 
\begin{align*}
    d_L(\mathcal{E}(u_i),\mathcal{E}(u_j))=d_L(\mathcal{E}_d(u_i),\mathcal{E}_d(u_j))+d_L(C_i,C_j)\geq 2t_f+1.
\end{align*}
By definition of FCLMC with data protection, the map $\mathcal{E}$ defines an $(f,t_d,t_f)_L$-FCLMC with data protection. Hence, $r_f^L(k,t_d,t_f)\leq N_L(D_2)+r_d$. This completes the proof.
\end{proof}

\begin{theorem}\label{thm-rnltf}
 Let $u_1,u_2,\dots,u_M\in\mathbb{Z}_m^k$ and $\mathcal{E}_d:\mathbb{Z}_m^k\to \mathbb{Z}_m^{k+r_d}$ be a data encoding map with redundancy $r_d$ and $t_d$-error-correction (that is, the corresponding code has the minimum Lee distance $2t_d+1$). Then 
 \begin{align*}
     N_L(D_{\mathcal{E}_d,f}(t_f,u_1,u_2,\dots,u_M))\leq N_L(M,2(t_f-t_d)),
 \end{align*}
 where $N_L(M,2(t_f-t_d))$ is the minimum possible length of error-correcting code with $M$ codewords and minimum Lee distance $2(t_f-t_d)$.
\end{theorem}
\begin{proof}
Since $d_L(\mathcal{E}_d(u_i),\mathcal{E}_d(u_j))\geq 2t_d+1$ for all distinct $u_i,u_j\in \mathbb{Z}_m^k$, 
\begin{align*}
    [D_{\mathcal{E}_d,f}(t_f,u_1,u_2,\dots,u_M)]_{ij}=&[2t_f+1-d_L(\mathcal{E}_d(u_i),\mathcal{E}_d(u_j))]^+\\
    \leq & 2t_f-2t_d=2(t_f-t_d).
\end{align*}
Hence, $ N_L(D_{\mathcal{E}_d,f}(t_f,u_1,u_2,\dots,u_M))\leq N_L(M,2(t_f-t_d))$.
\end{proof}

Combining Theorem \ref{thm-rnltd}, Theorem \ref{thm-rnltf}, and Theorem \ref{nl2t}, we have the following upper bound on the optimal redundancy.

\begin{corollary}
 Let $\mathcal{E}_d:\mathbb{Z}_m^k\to \mathbb{Z}_m^{k+r_d}$ be a data encoding map with redundancy $r_d$ and $t_d$-error correction (that is, the corresponding code has the minimum Lee distance $2t_d+1$). Then
 \begin{align*}
     r_f^L(k,t_d,t_f)\leq  m^k \left \lceil \frac{2(t_f-t_d)}{\left \lfloor \frac{m}{2}\right \rfloor}\right \rceil+r_d.
 \end{align*}
\end{corollary}

 For $m=2$, the Lee metric coincides with the Hamming metric over $\mathbb{Z}_2$.
\begin{corollary}
    Let $m=2$, and $\mathcal{E}_d:\mathbb{Z}_2^k\to \mathbb{Z}_2^{k+r_d}$ be an encoding map that defines a linear code with parameters $[n=k+r_d,k,2t_d+1]$. Then
    \begin{align*}
        r_f^L(k,t_d,t_f)\leq  2^{k+1}(t_f-t_d)+n-k. 
    \end{align*}
\end{corollary}

\subsection{A Plotkin-like bound for FCLMC with data protection}

A general Plotkin-type bound for function-correcting codes under the Lee metric was established in \cite[Lemma 4.1]{verma2025FCLMC}. In the previous subsection, we provided several estimates of the optimal redundancy in terms of the distance requirement matrices. However, applying this bound requires knowledge of all pairwise Lee distances between the codewords of an irregular Lee distance code. 
 To circumvent this requirement, we derive an alternate form of Plotkin-type lower bound on the optimal redundancy of FCLMCs with data protection that depends only on the prescribed distance parameters $d_d$ and $d_f$, together with the underlying function. It should be noted that the lower bound presented in Theorem \ref{rlower} applies only to functions whose image cardinality does not exceed $(m-1)k$.

\begin{theorem}[Lower Bound]
Let $f:\mathbb{Z}_m^k\to \text{Im}(f)$ be a function. Then, for any $(f,d_d,d_f)$-FCLMC over $\mathbb{Z}_m$, we have 
\begin{align*}
r_f^L(k,d_d,d_f)\geq \frac{(\ell-1)d_d+(m^k-\ell)d_f}{Sm^{k-1}}-k,
\end{align*}
where $\ell=\max_{e\in \mathrm{Im}(f)}|f^{-1}(e)|$, $d_f>d_d$, and 
\begin{align*}
    S=\begin{cases}
        \frac{m^2}{4} & \text{ if } m \text{ is even}\\
        \frac{m^2-1}{4}& \text{ if } m \text{ is odd}.\\
    \end{cases}
\end{align*}
\end{theorem}
\begin{proof}
Let $\mathcal{E}:\mathbb{Z}_m^k\to \mathbb{Z}_m^{k+r}$ be an encoding map that defines $(f,d_d,d_f)_L$-FCLMC with redundancy $r$. Then $C=\{\mathcal{E}(u):\ u\in \mathbb{Z}_m^k\}$ is an $(f,d_d,d_f)_L$-FCLMC of length $k+r$ with $m^k$ codewords. For a codeword $c\in C$, denote $u_c=f^{-1}(c)\in \mathbb{Z}_m^k$ and $f(u_c)=e_c\in \mathrm{Im}(f)$. Observe that $d_L(c,c')\geq d_d$ for all $c'\in C$ with $e_c=e_{c'}$ and $d_L(c,c'')\geq d_f$ for all $c''\in C$ with $e_c\neq e_{c''}$. Therefore,
\begin{align*}
    J(c)=\sum_{c'\in C, c\neq c'}d_L(c,c')&\geq (|f^{-1}(e_c)|-1)d_d+(m^k-|f^{-1}(e_c)|)d_f\\
    &=m^kd_f-d_d-(d_f-d_d)|f^{-1}(e_c)|.
\end{align*}
As $d_f>d_d$ and $\ell=\max_{e\in \mathrm{Im}(f)}|f^{-1}(e)|$, for all $c\in C$, we have
\begin{align*}
 J(c)\geq m^kd_f-d_d-(d_f-d_d)\ell=(\ell-1)d_d+(m^k-\ell)d_f.   
\end{align*}
Consequently, 
\begin{align}\label{eq1-plot}
 \sum_{(c,c')\in C^2, c\neq c'} d_L(c,c')\geq \sum_{c\in C} J(c)\geq m^k((\ell-1)d_d+(m^k-\ell)d_f).    
\end{align}
Also, from \cite[Eq. (15)]{Wyner1968}, we have 
    \begin{align}\label{eq2-plot}
        \sum_{(c,c')\in C^2, c\neq c'}d_L(c,c')&
        \leq Sm^{2k-1}(k+r),   
    \end{align}
where \begin{align*}
    S=\begin{cases}
        \frac{m^2}{4}, & \text{ if } m \text{ is even}\\
        \frac{m^2-1}{4}, & \text{ if } m \text{ is odd.}
    \end{cases}
\end{align*}
On combining Eqs. (\ref{eq1-plot}) and (\ref{eq2-plot}), we get 
\begin{align*}
    &Sm^{2k-1}(k+r)\geq m^k((\ell-1)d_d+(m^k-\ell)d_f),\\
    \text{i.e.,}\hspace{0.5 cm }& r\geq \frac{(\ell-1)d_d+(m^k-\ell)d_f}{Sm^{k-1}}-k.
\end{align*}
This completes the proof.
\end{proof}

\begin{corollary}    
Let $f:\mathbb{Z}_2^k\to \text{Im}(f)$ be a function. Then, for any $(f,d_d,d_f)$-FCLMC over $\mathbb{Z}_2$, we have 
\begin{align*}
r_f^L(k,t_d,t_f)\geq \frac{(\ell-1)d_d+(2^k-\ell)d_f}{2^{k-1}}-k,
\end{align*}
where $\ell=\max_{e\in \mathrm{Im}(f)}|f^{-1}(e)|$ and $d_f>d_d$.
\end{corollary}

\section{FCLMCs with data protection for specific functions}
In this section, we investigate the construction of function-correcting Lee metric codes (FCLMCs) with data protection for several important classes of functions. We derive general lower and upper bounds on the optimal redundancy and provide constructions of such codes.

\subsection{ Locally bounded functions}
\begin{definition}
    Let $u\in \mathbb{Z}_m^k$. A function ball centered at $u$ of radius $\rho$ is defined as
    $$B_L^f(u,\rho)=\{f(v)\mid v\in B_L(u,\rho)\}.$$
\end{definition}

\begin{definition}
    A function $f:\mathbb{Z}_m^k\to Im(f)$ is said to be a locally $(\lambda,\rho)_L$-function if $\mid B_L^f(u,\rho)\mid\leq \lambda$ for all $u\in \mathbb{Z}_m^k$.
\end{definition}
When $\lambda=2$, we say the function is a locally binary Lee function. Now, we derive an upper bound on the optimal redundancy for locally binary Lee functions. 

\begin{theorem}[Upper Bound]\label{locall-2r}
    Let $f:\mathbb{Z}_m^k\to \text{Im}(f)$ be a locally $(2,2t_f)_L$-function and $\mathcal{E}_d:\mathbb{Z}_m^k\to \mathbb{Z}_m^{k+r_d}$ be a data encoding map with redundancy $r_d$ and $t_d$-error-correction (that is, corresponding code has minimum Lee distance $2t_d+1$). Then for any $(f,t_d,t_f)$-FCLMC, we have 
    \begin{align*}
        r_f^L(k,t_d,t_f)\leq r_d+ \frac{2(t_f-t_d)}{\left \lfloor \frac{m}{2}\right \rfloor}.
    \end{align*}
\end{theorem}
\begin{proof}
 Denote    $N=\frac{2(t_f-t_d)}{\left \lfloor \frac{m}{2}\right \rfloor}$. Define an encoding map $\mathcal{E}:\mathbb{Z}_m^k\to \mathbb{Z}_m^{k+r_d+N}$ as follows:
 \begin{align*}
 \mathcal{E}(u)=(\mathcal{E}_d(u),C_u) \ \forall u\in \mathbb{Z}_m^k,    
 \end{align*}
 where 
 \begin{align*}
     C_u=\begin{cases}
        \underbrace{\left \lfloor \frac{m}{2}\right \rfloor, \left \lfloor \frac{m}{2}\right \rfloor, \dots, 
        \left \lfloor \frac{m}{2}\right \rfloor}_{N-times} & \text{ if } f(u)=\max B_L^f(u,2t_f),\\
        \underbrace{0,0,\dots,0}_{N-times} & \text{ otherwise}.
     \end{cases}
 \end{align*}
 We show that $\mathcal{E}$ defines an $(f,t_d,t_f)$-FCLMC with redundancy $r_d+N$.  For any $u,v\in \mathbb{Z}_m^k$ with $u\neq v$, we have 
 \begin{align*}
     d_L(\mathcal{E}(u),\mathcal{E}(v))=d_L(\mathcal{E}_d(u),\mathcal{E}_d(v))+d_L(C_u,C_v)\geq 2t_d+1.
 \end{align*}
 Now, let $u,v \in\mathbb{Z}_m^k $ such that $f(u)\neq f(v)$. Then we have the following two cases.
 \begin{enumerate}
     \item If $d_L(u,v)\geq 2t_f+1$, then 
     \begin{align*}
         d_L(\mathcal{E}(u),\mathcal{E}(v))=d_L(\mathcal{E}_d(u),\mathcal{E}_d(v))+d_L(C_u,C_v)\geq 2t_f+1.
     \end{align*}
     \item If $d_L(u,v)\leq 2t_f$, then $f(v)\in B_L^f(u,2t_f)$ and we have $d_L(C_u,C_v)=N\left \lfloor \frac{m}{2}\right \rfloor=2(t_f-t_d)$. Thus
     \begin{align*}
        d_L(\mathcal{E}(u),\mathcal{E}(v))=d_L(\mathcal{E}_d(u),\mathcal{E}_d(v))+d_L(C_u,C_v)\geq 2t_d+1+N \left \lfloor \frac{m}{2}\right \rfloor=2t_f+1.  
     \end{align*}
 \end{enumerate}
\end{proof}

\begin{corollary}
    Let $\mathcal{E}_d:\mathbb{Z}_2^k \to \mathbb{Z}_2^{k+r_d}$ be a data encoding that gives a systematic linear code with parameters $[n=k+r_d,k,2t+d+1]$. Then for any locally $(2,2t_f)$-function, we have 
    \begin{align*}
        r_f^H(k,t_d,t_f)\leq n-k+2(t_f-t_d).
    \end{align*}
\end{corollary}
\begin{proof}
    We get the desired result by putting $m=2$ in Theorem \ref{locall-2r}.
\end{proof}

\subsection{ The Lee weight function}
The Lee weight function $f:\mathbb{Z}_m^k\to \{0,1,2,\dots,k \left \lfloor \frac{m}{2}\right \rfloor\}$ is defined as $f(u)=w_L(u)$ for all $u\in\mathbb{Z}_m^k$. We derive bounds on the optimal redundancy. Observe that $|\mathrm{Im}({f})|=k\lfloor \frac{m}{2}\rfloor+1\leq (m-1)k$ for $m,k\geq 3$.  We have a lower bound on the optimal redundancy from Theorem \ref{rlower}.
\begin{corollary}
 Let $f(u)=w_L(u)$ for all $u\in \mathbb{Z}_m^k$ be the Lee weight function. For $m\geq 3$ and $k\geq 3$, we have
 \begin{align*}
     r_f^L(k,t_d,t_f)\geq \left \lceil \frac{4t_f+2t_d-2m+3}{m}\right \rceil.
 \end{align*}
\end{corollary}

Next, we present an upper bound on the optimal redundancy.

\begin{theorem}[Upper Bound]\label{lee-w-r}
    Let $f:\mathbb{Z}_m^k\to \{0,1,\dots,k\left \lfloor \frac{m}{2}\right \rfloor\}$ be the Lee weight function defined by $f(u)=w_L(u)$ for all $u\in \mathbb{Z}_m^k$, and  $\mathcal{E}_d:\mathbb{Z}_m^k\to \mathbb{Z}_m^{k+r_d}$ be a systematic data encoding map (first $k$ coordinates correspond to message) with $t_d$-error-correction. Then
    $$r_f^L(k,t_d,t_f)\leq r_d+N_L(2t_f+1,2(t_f-t_d)),$$
    where $N_L(\lambda,d)$ is the minimum length of an error-correcting code with $\lambda$ codewords and minimum Lee distance $d$.
\end{theorem}
\begin{proof}
Denote $n_1=k+r_d$ and $n_2=N_L(2t_f+1,2(t_f-t_d))$. Let $C$ be a code  with $2t_f+1$ codewords, length $n_2$ and minimum Lee distance $2(t_f-t_d)$. Denote the codewords of $C$  by   $C_0,C_1,\dots, C_{2t_f}$. Define an encoding 
$\mathcal{E}:\mathbb{Z}_m^k\to \mathbb{Z}_m^{n_1+n_2}$ as follows
\begin{align*}
    \mathcal{E}(u)=(\mathcal{E}_d(u),C_{f(u)\pmod{2t_f+1}}).
\end{align*}
We show that $\mathcal{E}$ defines an $(f,t_d,t_f)$-FCLMC. For $u,v\in \mathbb{Z}_m^k$ with $u\neq v$, we have 
\begin{align*}
    d_L(\mathcal{E}(u),\mathcal{E}(v)= d_L(\mathcal{E}_d(u),\mathcal{E}_d(v)+d_L(C_{f(u)\pmod{2t_f+1}},C_{f(v)\pmod{2t_f+1}})\geq 2t_d+1.
\end{align*}
Next, let $u,v\in \mathbb{Z}_m^k$ with $f(u)\neq f(v)$. \\
\textbf{Case 1} If $|f(u)-f(v)|\leq 2t_f$, then $C_{f(u)\pmod{2t_f+1}}\neq C_{f(v)\pmod{2t_f+1}}$. Thus, we have 
\begin{align*}
    d_L(\mathcal{E}(u),\mathcal{E}(v))=& d_L(\mathcal{E}_d(u),\mathcal{E}_d(v))+d_L(C_{f(u)\pmod{2t_f+1}},C_{f(v)\pmod{2t_f+1}})\\ \geq& 2t_d+1 +2(t_f-t_d)=2t_f+1.
\end{align*}
\textbf{Case 2} If $|f(u)-f(v)|\geq  2t_f+1$, then without loss of generality, we assume $f(u)>f(v)$, we get $f(u)-f(v)\geq 2t_f+1$. Also, $d_L(u,v)\geq w_L(u)-w_L(v)\geq 2t_f+1$. Hence,
\begin{align*}
    d_L(\mathcal{E}(u),\mathcal{E}(v))=& d_L(\mathcal{E}_d(u),\mathcal{E}_d(v))+d_L(C_{f(u)\pmod{2t_f+1}},C_{f(v)\pmod{2t_f+1}})\\ \geq& 2t_f+1.
\end{align*}
\end{proof}

\begin{corollary}
Let $f:\mathbb{Z}_2^k\to \{0,1,\dots,k\}$ be the Hamming weight function defined by $f(u)=w_H(u)$ for all $u\in \mathbb{Z}_2^k$, and  $\mathcal{E}_d:\mathbb{Z}_m^k\to \mathbb{Z}_m^{k+r_d}$ be a systematic data encoding map with $t_d$-error-correction that gives a linear code with parameters $[n=k+r_d,k,2t_d+1]$. Then
    $$r_f^L(k,t_d,t_f)\leq n-k+N_L(2t_f+1,2(t_f-t_d)).$$   
\end{corollary}
\begin{proof}
For $m=2$, the Lee weight and the Hamming weight coincide over $\mathbb{Z}_m^k$. The rest follows from Proposition \ref{lee-w-r}.   
\end{proof}

\subsection{Modular sum function}
A modular sum function $f_{\mathrm{ms}}:\mathbb{Z}_{m}^k\to \mathbb{Z}_{m}$ is defined as $f_{\mathrm{ms}}(u)=\sum_{i=1}^ku_i \pmod{m}$ for all $u=(u_1,u_2,\dots,u_k)\in \mathbb{Z}_{m}^k$. For the modular sum function $f_{\mathrm{ms}}$, we have $|\mathrm{Im}(f_{\mathrm{ms}})|=m$. Therefore, from Theorem \ref{rlower}, we have a lower bound on the optimal redundancy.

\begin{corollary}
 Let $f_{\mathrm{ms}}:\mathbb{Z}_{m}^k\to\mathbb{Z}_{m}$ be the modular sum function. Then 
 \begin{align*}
     r_{f_{\mathrm{ms}}}^L(k,t_d,t_f)\geq \left \lceil \frac{4t_f+2t_d-2m+3}{m}\right \rceil.
 \end{align*}
\end{corollary}

We now present an upper bound on the optimal redundancy.
\begin{theorem}[Upper Bound]
  Let $f_{\mathrm{ms}}:\mathbb{Z}_{m}^k\to\mathbb{Z}_{m}$ be the modular sum function, and let $\mathcal{E}_d:\mathbb{Z}_m^k\to \mathbb{Z}_m^{k+r_d}$ be a data encoding map with $t_d$-error-correction. Then $$r_{f_{\mathrm{ms}}}^L(k,t_d,t_f)\leq  r_d+ N_L(m,2(t_f-t_d)).$$   
\end{theorem}
\begin{proof}
 Let $n=N_L(m,2(t_f-t_d))$ and $C=(C_1,C_2,\dots,C_m)$ be a code of length $n$ with the minimum Lee distance $2(t_f-t_d)$. Define an encoding $\mathcal{E}:\mathbb{Z}_m^k\to \mathbb{Z}_m^{k+r_d+n}$ as follows:
 \begin{align*}
     \mathcal{E}(u)=(\mathcal{E}_d(u),C_i) \text{ whenever } f_{\mathrm{ms}}(u)=i.
 \end{align*}
 It is easy to see that 
 \begin{align*}
     d_L(\mathcal{E}(u),\mathcal{E}(v))\geq d_L(\mathcal{E}_d(u),\mathcal{E}_d(v))\geq 2t_d+1 \text{ for all } u\neq v
 \end{align*}
 and 
 \begin{align*}
     d_L(\mathcal{E}(u),\mathcal{E}(v))\geq d_L(\mathcal{E}_d(u),\mathcal{E}_d(v))+d_L(C_i,C_j) \geq 2t_d+1+2(t_f-t_d)=2t_f+1
 \end{align*}
 for all $u,v\in \mathbb{Z}_m^k$ with $i=f_{\mathrm{ms}}(u)\neq f_{\mathrm{ms}}(v)=j$.
 Thus, $\mathcal{E}$ defines an $(f,t_d,t_f)_L$-FCLMC. This completes the proof.
\end{proof}

\section{Lee distance graphs and strict $(f,d_d,d_f)_L$-FCLMCs}
Note that an $(f,d_d,d_f)_L$-FCLMC with $d_d=d_f$ reduces to an error-correcting code in the Lee metric. We say an $(f,d_d,d_f)_L$-FCLMC, a strict $(f,d_d,d_f)_L$-FCLMC whenever $d_f>d_d$. In this section, we investigate the existence and non-existence of strict $(f,d_d,d_f)_L$-FCLMCs. First, we define the minimum Lee distance graph corresponding to a given code over $\mathbb{Z}_m$ under the Lee metric.

\begin{definition}[ $\alpha$-Lee distance graph]
 Let $C$ be a code over $\mathbb{Z}_m$ with the minimum Lee distance $d_L(C)=d$. For an positive integer $\alpha$, the $\alpha$-Lee distance graph of $C$, denoted by $G_C^{\alpha}$, is a graph with vertex set $C$ and two distinct vertices $c_1,c_2\in C$ are adjacent if and only if $d_L(c_1,c_2)\leq \alpha$. 
\end{definition}

\begin{theorem}
Let $C$ be an $(n, m^k,d_L(C)=d)$ code over $\mathbb{Z}_m$ in the Lee metric and $f:\mathbb{Z}_m^k\to \mathrm{Im}(f)$ be a function with $|\mathrm{Im}(f)|\geq 2$. If $C$ is a strict $(f,d_d=d,d_f)_L$-FCLMC for some $d_f>d$, then the $\alpha$-Lee distance graph $G_C^{\alpha}$ of $C$ can not be connected for any $d\leq \alpha< d_f$.   
\end{theorem}
\begin{proof}
For all $x\in \mathrm{Im}(f)$, define $x_C=\{c_u\in C:\ f(u
)=x\}\subseteq C$, where $c_u$ is the encoded codeword corresponding to the message vector $u\in\mathbb{Z}_m^k$. Let $C$  is a strict $(f,d_d=d,d_f)_L$-FCLMC for some $d_f>d$. For  each $c_u\in x_C$ and each $c_v\in y_C$ with $x\neq y$, by definition of $(f,d_d=d,d_f)_L$-FCLMC, we have
\begin{align*}
    d_L(c_u,c_v)\geq d_f>d.
\end{align*}
Recall that in $G_C^{\alpha}$, two vertices are adjacent if and only if the Lee distance between them is at most $\alpha$.
Therefore, there is no edge in $G_C^{\alpha}$ for $\alpha<d_f$ between $x_C$ and $y_C$ for distinct $x$ and $y$. Thus, $G_C^{\alpha}$ is not connected for $\alpha<d_f$. 
\end{proof}

The above theorem can be generalized to the codes whose
$\alpha$-Lee distance graph has more than one connected component.

\begin{theorem}
Let $C$ be an $(n,m^k,d)$ code over $\mathbb{Z}_m^k$ in the Lee metric. If the $\alpha$-Lee distance graph $G_C^{\alpha}$ of $C$ has $\eta$ connected components, then $C$ cannot be a strict $(f,d,d_f)_L$-FCLMC  for any function $f:\mathbb{Z}_m^k\to \mathrm{Im}(f) $   with $|\mathrm{Im}(f)|\geq \eta+1$ and  $\alpha<d_f$. 
\end{theorem}

In the following theorem, we discuss the converse direction that gives an existence criterion for arbitrary codes to be strict FCLMCs.
\begin{theorem}
Let $C$ be an $(n,m^k,d)$  code. Let $f:\mathbb{Z}_m^k\to \mathrm{Im}(f)$ be a function with $\mathrm{Im}(f)=\{e_1,e_2,\dots,e_{\eta}\}$. For any integer $\alpha\geq d$, if connected components of $G_C^{\alpha}$ can be grouped into disjoint unions of $C_1,C_2,\dots, C_{\eta}$ with $|C_i|=|f^{-1}(e_i)|$, then $C$ is a strict $(f,d_d=d, d_f=\alpha+1)_L$-FCLMC.   
\end{theorem}
\begin{proof}
Since  $|C_i|=|f^{-1}(e_i)|$ for all $1\leq i\leq \eta$, let $\psi_i:f^{-1}(e_i)\to C_i$ be a bijective map. For $u\in \mathbb{Z}_m^k$, define an encoding $\mathcal{E}(u)=\psi_i(u)$ whenever $f(u)=e_i$. Then, for distinct $u, v\in \mathbb{Z}_m^k$,
\begin{align*}
    d_L(\mathcal{E}(u),\mathcal{E}(v))=&d_L(\psi_i(u),\psi_j(v)) \text{ for some } 1\leq i,j\leq \eta\\
    \geq& d_d, 
\end{align*}
since $\psi_i(u),\psi_j(v)\in C$ and $d_L(C)=d_d$.
Furthermore, let $u,v\in \mathbb{Z}_m^k$ with $f(u)\neq f(v)$, that is, $f(u)=e_i$ and  $f(v)=e_j$ for some $i\neq j$. Then, $\mathcal{E}(u)=\psi_i(u)\in C_i$ and $\mathcal{E}(v)=\psi_j(v)\in C_j$ are in distinct connected components of $G^{\alpha}_C$. Thus,
\begin{align*}
   d_L(\mathcal{E}(u),\mathcal{E}(v))
    \geq& \alpha+1=d_f.   
\end{align*}
Hence, by Definition \ref{def-fclmc}, $C$ is a strict $(f,d_d=d, d_f=\alpha+1)_L$-FCLMC. 
\end{proof}

\section{Conclusion}
In this work, we proposed a systematic framework to study function-correcting codes in the Lee metric that simultaneously protects both function values and the corresponding message data. We established several bounds, including a Plotkin-like bound on the optimal redundancy of these codes, as well as explicit constructions. Using these general results, we derived explicit bounds for specific classes of functions, the Lee weight function, locally bounded functions in the Lee metric, and modular sum functions. Furthermore, we explored strict function-correcting Lee metric codes using the corresponding Lee distance graph.

This work lays a foundation for the study of function-correcting codes with data protection in the Lee metric as part of an ongoing research project. The research is currently in progress, with further efforts focused on establishing additional structural properties, deriving stronger bounds, and developing related constructions. A comprehensive treatment of these results will be incorporated in a subsequent version of this work.


\bibliographystyle{abbrv}
	\bibliography{ref}

\end{document}